\documentclass[12pt,reqno]{amsart}
\usepackage{dsfont, amssymb,amsmath,amscd,latexsym, amsthm, amsxtra,amsfonts}
\usepackage[all]{xy}
\usepackage{graphicx}
\usepackage{float}
\usepackage{txfonts}
\newtheorem{theorem}{Theorem}[section]
\newtheorem{lemma}[theorem]{Lemma}

\newtheorem{example}[theorem]{Example}

\newtheorem{LA}{Lemma A }
\newtheorem{CA}{Corollary A}

\newtheorem{Them}{Theorem}[section]
\newtheorem{Prop}{Proposition}[section]
\newtheorem{Remark}{Remark}[section]

\begin{document}
\makeatletter
\def\@setauthors{%
\begingroup
\def\thanks{\protect\thanks@warning}%
\trivlist \centering\footnotesize \@topsep30\p@\relax
\advance\@topsep by -\baselineskip
\item\relax
\author@andify\authors
\def\\{\protect\linebreak}%
{\authors}%
\ifx\@empty\contribs \else ,\penalty-3 \space \@setcontribs
\@closetoccontribs \fi
\endtrivlist
\endgroup } \makeatother
 \baselineskip 17pt
\title[{{  Optimal  Control of Insurance Company   }}]
 {{\small Theoretical and numerical Analysis on Optimal
  dividend policy of an  insurance  company
  with positive transaction  cost  and  higher solvency   }}
\author[{\bf Z.Liang and J. Yao } ]
{ Zongxia Liang \\ Department of Mathematical Sciences, Tsinghua
University, Beijing 100084, China. Email:
zliang@math.tsinghua.edu.cn  \\Jicheng Yao \\ Department of
Mathematical Sciences, Tsinghua University, Beijing 100084, China.
Email: yaojicheng626@163.com}
\begin{abstract}
Based on a point of view that solvency and security are first, this
paper considers regular-singular stochastic optimal control problem
of a large insurance company facing positive transaction cost asked
by reinsurer under solvency constraint. The company controls
proportional reinsurance and dividend pay-out policy to maximize the
expected present value of the  dividend pay-outs until the  time of
bankruptcy. The paper aims at deriving  the optimal retention ratio,
dividend payout level, explicit value function of the insurance
company via stochastic analysis and PDE methods.  The results
present the best equilibrium point  between maximization of dividend
pay-outs and  minimization of risks. The paper also gets a
risk-based capital standard to ensure the capital requirement of can
cover the total given risk. We present numerical results to make
analysis how the model parameters, such as, volatility, premium
rate, and risk level, impact on risk-based capital standard, optimal
retention ratio, optimal dividend payout level and the company's
profit.  \vskip 10pt \noindent
 \end{abstract}
\maketitle \setcounter{equation}{0}
\section{{\small {\bf Introduction}}}
 \vskip 10pt \noindent
In this paper we consider a problem of risk control and dividend
optimization for a large insurance company facing positive
transaction cost asked by reinsurer( that is, the case of
excess-of-loss reinsurance). The company controls dividend stream
and its risk, as well as potential profit by choosing different
business activities among all of available policies to it. The
objective of the insurer is to choose proportional reinsurance
  and dividend level to maximize the expected present value of the
  dividend pay-outs until the  time of bankruptcy. This is a
regular-singular control problem of diffusion processes.  In the
view of optimization of the dividend pay-outs, the stochastic
optimal control problems of a large insurance company have been
given attention by many authors recently. We refer the readers to
Taksar and Zhou\cite{ime04}(1998), Choulli, Taksar and
Zhou\cite{CTZ}(2001), H{\o}jgaard and Taksar\cite{s3,s4}(1999,
2001), Asmussen et all\cite{s7,s8}(1997,2000), Guo, Liu and Zhou
\cite{s1}(2004), He and Liang\cite{ime01, ime03}(2008) and other
authors' works. According to classical economic theory, the approach
used in some of these papers is  the insurer selects  one from all
admissible business arrangements to yield maximization of expected
present value of dividend pay-outs.  However, Although this ideal
approach is the best in concept, it can't be used in practice
because the insurance business is a business {\sl affected with a
public interest}  and consumers should be protected against insurer
insolvencies (cf.Chapter 34, Williams and Heins\cite{law}(1985),
Riegel and Miller \cite{law2}(1963), Welson and Taylor
\cite{law3}(1958) ). Therefore, a policy making the company go
bankrupt before termination of contract between insurer and policy
holders or a policy of  low solvency(where {\sl solvency} means {\sl
$1-$ probability of bankruptcy}, cf.Bowers, Gerber et all
\cite{law1}(1997)) does not seem to be the best way and should be
prohibited even though it has the highest gain because under which
no claims will return to policy holders and contract can not be
forced to perform. On the other hand, this policy will also worsen
basis of insurance business survival and
 company's reputation which is the company's chief asset- a plant of long
growth but peculiarly susceptible to the cold winds of idle rumor.
So the higher standard of security and solvency is the first factor
to be taken into account for insurer. \vskip 10pt \noindent
Unfortunately, there are very few results concerning on  stochastic
optimal control problems of  insurance company
 from a view of security and solvency are first. Paulsen
\cite{pau}(2003) first studied this kinds of  optimal controls for
diffusions via properties of return function and then  He, Hou and
Liang\cite{ime02}(2008) investigated the optimal control problems
for linear Brownian model in case of cheaper reinsurance.
 By an innovative idea, based on a point of view that
security and solvency are first, in this paper we  will  establish a
sophisticated setting to effectively solve this kind of optimal
control on problems of a large insurance company in case of positive
transaction cost and solvency constraint. We aim at deriving  the
optimal retention ratio, dividend payout level, explicit value
function of the insurance company via stochastic analysis and PDE
methods. The model treated and approach used in the present paper
are different from those of \cite{pau}. In our approach, only
admissible policies satisfying this standard of security are
considered, so it will reduce the insurer's expected present value
of dividend pay-outs, on the other hand, it will increase security
and solvency in some sense by minimal loss. From this set of
admissible policies, the insurer can select one that allows the
highest expected present value of dividend pay-outs. Indeed, Our
results present the best place between gains and risks, the loss for
higher security and solvency is minimal. To get these results we
first study some properties of probability of bankruptcy by
stochastic analysis and PDE methods, then solve a generalized HJB
equations in appendix , finally we prove that solution of the HJB is
the optimal return function of the company. We find that the case
treated in the \cite{ime02}(2008) is a trivial case, that is, the
company of the model in the \cite{ime02}(2008) will never go to
bankruptcy, it is an ideal  model in concept, and it indeed does not
exist in reality.  Because  probability of bankruptcy for the model
treated in the present paper  is very large,  our results can not be
directly deduced from the \cite{ime02}(2008).
 \vskip 10pt \noindent
 The paper is organized as follows. In next section  we
 establish  mathematical model of a large insurance
company treated in this paper. In section 3 we present
  main result of this paper and its
economic interpretations. In section 4 we analyze solvency and
security of  stochastic mathematical model considered in this paper,
the results in this section also state that the solvency constraint
set $ \mathfrak{B}  $  in section 3 is not empty set nor $R_+$, so
the setting treated in this paper is well defined. In section 5 we
present numerical results studying how the model parameters impact
on the optimal return function and dividend policy. In section 6 we
list some lemmas of properties  of bankrupt probability, and their
rigorous proofs are presented in section 8. We give detailed proofs
of main results of this paper in section 7. Optimal return function
and its robustness properties w.r.t. dividend level are given in
appendix.
 \vskip 10pt \noindent
\setcounter{equation}{0}
\section{{\small {\bf  Mathematical model of a large insurance
company  }}}
 \vskip 10pt\noindent
 We start with a filtered probability space $(\Omega, \mathcal {F}, \{ \mathcal
{F}_{t}\}_{t\geq 0}, \mathbf{P})$ with a standard Brownian motion
$\{W_t\}_{t\geq 0}$ on it, adapted to the filtration $\mathcal
{F}_{t}$ satisfying the usual conditions. A pair of $\mathcal
{F}_{t}$ adapted processes $\pi =  \{ a_\pi(t), L^\pi_t\}$ is called
 a admissible policy if $0\leq a_\pi(t)\leq 1 $ and $L^\pi_t\ $ is a
 nonnegative,  non-decreasing, right-continuous with left limits.
  We denote by
$\Pi$ the whole set of  admissible policies. \vskip 10pt\noindent
Given an  admissible policy $ \pi $, if we denote by $R^\pi_{t}$ the
reserve of a large insurance company at time $t$ and by $L^\pi_t$
cumulative amount of dividends paid out to the shareholders up to
time $t$, then, by using the center limit theorem, we can assume
that (see \cite{CTZ,ime04,s7,C1, C2, C3, C4}) the dynamics of $R^\pi_{t}$ is given
by
\begin{eqnarray}\label{21}
dR^\pi_{t}=(\mu-(1-a_\pi(t))\lambda)dt+\sigma a_\pi(t)d
{W}_{t}-dL^\pi_{t}, \quad R^\pi_{0}=x,
\end{eqnarray}
where  $1-a(t)$ is the reinsurance  fraction at time $t$, the $
R^\pi_{0}=x$ means that the initial liquid reserve is $x$, the
constants $\mu$ and $\lambda$ can be regarded as the safety loadings
of the insurer and reinsurer, respectively. \vskip 10pt\noindent
Throughout this paper we assume that transaction cost $ \lambda-\mu
>0$. We refer readers to He, Hou and
Liang\cite{ime02}(2008) for $\lambda=\mu$. When the reserve
vanishes, we say that the company is bankrupt.  We define the time
of bankruptcy by $\tau ^\pi _x=\inf\big \{t\geq0: R^\pi_{t}= 0\big
\}$. Obviously, $\tau ^\pi _x$ is an $\mathcal {F}_{t}$ -stopping
time. For any $b\geq 0$, let $\Pi_b=\{\pi\in \Pi : \int _0 ^{ \infty
}I_{\{s:R^\pi(s)<b\}}dL_s^{\pi}=0\} $. It is easy to see that
$\Pi=\Pi_0$ and $b_1>b_2\Rightarrow \Pi_{b_1}\subset \Pi_{b_2}$.
\vskip 10pt\noindent
 For a given admissible policy $\pi$
we define  value function $V(x)$ of a large insurance company by
\begin{eqnarray}\label{f3}
J(x,\pi)&=&\mathbb{E}\big \{\int_0^{\tau^\pi_x } e^{-ct} dL_t^\pi\big\},\nonumber \\
 V(x,b )&=&\sup_{\pi \in \Pi_b}\{J(x,\pi)\},
\end{eqnarray}
\begin{eqnarray}\label{f4}
 V(x)&=&\sup_{b \in \mathfrak{B}}\{  V(x,b )  \}
\end{eqnarray}
where the solvency set $\mathfrak{B} $ defined by
  \begin{eqnarray*}
 \mathfrak{B}:=\big \{b\ :\ \mathbb{P}[\tau_{b}^{\pi_{b}} \leq T] \leq
\varepsilon \ , \  J(x, \pi_b)= V(x,b) \mbox{ and} \ \pi_b \in
\Pi_b\big\},
  \end{eqnarray*}
 $c>0$ is a discount rate, $\tau_b^{\pi_b}$ is
 the time of bankruptcy $\tau_x^{\pi_b} $  when
the initial asset $x=b$ and the control policy is $\pi_b$.
$1-\varepsilon$ is the standard of security and less than solvency
for given $\varepsilon>0 $( see \cite{law0, law1}). \vskip
10pt\noindent {The main purpose of this paper is to solve the
optimal control problems (\ref{f3}) and (\ref{f4}).} In addition to
finding optimal return function
 $V(x)$  of the company, we also derive  the optimal retention ratio,
dividend payout level $b^*$ and optimal policy $ \pi^*_{b^*}  $
associated with the $V(x)$ such that $J(x,\pi^*_{b^*})= V(x)$.
Moreover, their robustness properties w.r.t. model parameters are
presented via numerical results.
 \vskip 10pt \noindent
  \setcounter{equation}{0}
\section{{\small {\bf Main Results }} }
 \vskip 10pt \noindent
 In this section we first present main results of this paper, then,
 together with numerical
 results in section 5 below,  give
 economic and financial interpretations of the main results.
 The results present the best equilibrium point
 between benefits and risks. The proofs of main results
  will be given in section 7.
\vskip 5pt \noindent
\begin{Them}\label{theorem33}  Assume that
 transaction cost $\lambda-\mu>0$. Let  level of risk $ \varepsilon \in (0,
1)$ and time horizon $T$ be given. \vskip 5pt \noindent (i) If $
\mathbf{P}[\tau_{b_0}^{\pi^*_{b_0}}\leq T]\leq \varepsilon $,  then
the value function $ V(x) $ of the company  is $f(x)$ defined by
(\ref{32}) and (\ref{33}) in appendix, and $V(x)=V(x,
b_0)=J(x,\pi^*_{b_0})=V(x,0)=f(x)$.  The optimal policy associated
with $V(x) $ is $\pi_{b_o}^\ast=\{
 A^*_{b_0}(R^{\pi_{b_o}^\ast}_\cdot),L^{\pi_{b_o}^\ast}_\cdot\}$,
where  $(R^{\pi_{b_0}^\ast}_t,   L_t^{\pi_{b_0}^\ast  })$  is uniquely determined
by the
following SDE with reflection boundary(cf.\cite{s10100}):
\begin{eqnarray}\label{312}
\left\{
\begin{array}{l l l}
dR_t^{\pi_{b_o}^\ast}=(\mu-(1-A^*_{b_0}(R^{\pi_{b_o}^\ast}_t)
)\lambda)dt+\sigma A^*_{b_0}(R^{\pi_{b_o}^\ast}_t)d
{W}_{t}-dL_t^{\pi_{b_o}^\ast},\\
R_0^{\pi_{b_o}^\ast}=x,\\
0\leq R^{\pi_{b_o}^\ast}_t\leq  b_0,\\
\int^{\infty}_0 I_{\{t: R^{\pi_{b_o}^\ast}_t
<b_0\}}(t)dL_t^{\pi_{b_o}^\ast}=0
\end{array}\right.
\end{eqnarray}
and  $\tau_{x}^{\pi^*_{b_0}}=\inf\{t:R^{\pi_{b_0}^\ast}_t=0 \}$.
The optimal
dividend
 level is  $b_0$(see Lemma A.1 ),
  where $A^*(x)$ is defined by part (iii) of Lemma
 A.1 in appendix. The solvency of the company
is bigger than $1-\varepsilon$. \vskip 5pt \noindent (ii) If  $
\mathbf{P}[\tau_{b_0}^{\pi^*_{b_0}}\leq T]>\varepsilon $, then there
is a unique $ b^*> b_0$ satisfying
$\mathbf{P}[\tau_{b^*}^{\pi^*_{b^*}}\leq T]= \varepsilon $  such
that  $g(x)$  defined by (\ref{37}) and (\ref{38}) in appendix  is
the value function of the company, that is,
\begin{eqnarray}\label{313}
g(x)=\sup_{b\in \mathfrak{B}}\{V(x,b)\}=V(x, b^*)=J(x,\pi^*_{b^*})
\end{eqnarray}
 and
\begin{eqnarray}\label{314}
b^* \in \mathfrak{B},
\end{eqnarray}
where
\begin{eqnarray*}
\mathfrak{B}:=\big \{b: \mathbb{P}[\tau_{b}^{\pi_{b}} \leq T] \leq
\varepsilon, \  J(x, \pi_b)= V(x,b) \mbox{ and} \ \pi_b \in \Pi_b\
\big\}.
 \end{eqnarray*}
The optimal policy associated with $g(x)$ is
$\pi_{b^*}^\ast=\{
 A^*_{b^*}(R^{\pi_{b^*}^\ast}_\cdot),L^{\pi_{b^*}^\ast}_\cdot\}$,
 where $ ( R^{\pi_{b^*}^\ast}_\cdot, L^{\pi_{b^*}^\ast}_\cdot\}     )     $
is uniquely determined by  the following SDE with reflection boundary:
\begin{eqnarray}\label{315}
\left\{
\begin{array}{l l l}
dR_t^{\pi_{b^*}^\ast}=(\mu-(1- A^*_{b^*}(R^{\pi_{b^*}^\ast}_t)
)\lambda)dt+\sigma  A^*_{b^*}(R^{\pi_{b^*}^\ast}_t)d
{W}_{t}-dL_t^{\pi_{b^*}^\ast},\\
R_0^{\pi_{b^*}^\ast}=x,\\
0\leq R^{\pi_{b^*}^\ast}_t\leq  b^*,\\
\int^{\infty}_0 I_{\{t: R^{\pi_{b^*}^\ast}_t
<b^*\}}(t)dL_t^{\pi_{b^*}^\ast}=0
\end{array}\right.
\end{eqnarray}
and $\tau_{x}^{\pi^*_{b  }}=\inf\{t:R^{\pi_{b}^\ast}_t=0 \}$.
The optimal dividend
 level is  $b^*$, where $ A^*_{b^*}(x)$ is defined by part (iii) of Lemma
 A.2 in appendix. The optimal
 dividend  policy $\pi^*_{b^*}$ and
the optimal dividend $b^*$ ensure that the solvency of the company
is $1-\varepsilon$.
\vskip 5pt \noindent
(iii) Moreover,
\begin{eqnarray}\label{316}
 \frac{ g(x, b^*)  }{g(x,b_0)   }\leq 1.
 \end{eqnarray}
\end{Them}
\vskip 1cm\noindent
{\sl  We give economic and financial explanation of
Theorem \ref{theorem33} is as follows:}
 \vskip 10pt\noindent
{\sl  (1) For a given level of risk  and time horizon, if
probability of
 bankruptcy is less  than the level of risk, the optimal
  control problem of (\ref{f3}) and (\ref{f4}) is the traditional one,
the company has
higher solvency, so it will have  good reputation. The solvency
constraints here do not work.
 \vskip 5pt\noindent
 (2) If probability of bankruptcy is
large than the level of risk, the traditional optimal policy will
not meet the standard of security and solvency, the company needs to
find a sub-optimal policy $\pi_{b^*}^\ast $ to improve its solvency.
The sub-optimal reserve process $ R^{\pi_{b^*}^\ast}_t $ is a
diffusion process reflected at $b^*$, the process
$L^{\pi_{b^*}^\ast}_t  $ is the process which ensures the
reflection. The sub-optimal action is to pay out everything in
excess of $b^*$ as dividend and pay no dividend when the reserve is
below $b^*$, and $ A^*_{b^*}$ is the sub-optimal feedback control
function.\vskip 10pt\noindent
 (3) On the one hand, the inequality
(\ref{316}) states that $\pi_{b^*}^\ast $ will reduce the company's
profit, on the other hand, in the view of (\ref{316}),
$\mathbb{P}[\tau _{b^\ast}^{\pi_{b^*}^\ast}\leq T]= \varepsilon $ and
 Corollary A\ref{corollary32} below, the cost of improving solvency is
minimal. Therefore the policy $\pi_{b^*}^\ast $ is the best
equilibrium action between making profit and improving
solvency.\vskip 10pt\noindent (4) The under writing risk $\sigma^2$,
the premium rate $\mu $ and the initial capital $x$ will increase
the company's return, see the graphs \ref{difsigma} and \ref{difmu}
in section 5 below. \vskip 10pt\noindent (5) The risk-based capital
standard  $x$  decreases with the preferred risk level
$\varepsilon$, so  the higher preferred risk level $\varepsilon$
only needs  a lower  initial risk-based capital, see the graph
\ref{ep-x} below. \vskip 10pt\noindent (6)  The optimal dividend
level $ b(\varepsilon )$ is a decreasing function of the risk level
$ \varepsilon $, by comparison theorem of SDE, the optimal retention
ratio decreases with $ \varepsilon $, but the optimal dividend
process increases with $ \varepsilon $. Inversely, the risk level
$\varepsilon(b) $ is also
 a decreasing function of $b$ (see the graphs \ref{xinep-b}
and \ref{xinb-ep} below ).
 } \vskip 10pt\noindent
\begin{Remark}
Because the  \cite{ime02} had no continuity of probability of
bankruptcy and actual $b^*  $, the authors of \cite{ime02}  did not
obtain the best equilibrium policy $\pi_{b^*}^\ast $.
\end{Remark}
 \vskip 10pt \noindent
 \setcounter{equation}{0}
\section{{\small {\bf Analysis on the security and solvency of control model}}}
\vskip 10pt \noindent
 In this section we will give a quantitative analysis about
 the security and solvency of stochastic control
 model treated in  this paper.
  The main result of this section is Theorems \ref{theorem41}
  and \ref{lemma51} below. They
 reveal that for any given $T>0 $ low dividend $b_0$ will raise
level  of risk $\varepsilon_0$, the company does have higher level
of risk before the contract between insurer and policy holder goes
into effect(i.e., $T $ is less than the time of the contract issue
and positive), the company's solvency is less than
$1-\varepsilon_0$, so the company has to find an optimal dividend
policy that improves the ability of the insurer to fulfill its
obligation to policy holders under higher standard of security and
solvency. On the other hand,  the solvency constraint set $
\mathfrak{B} $ in section 3 is not empty set nor $R_+$,  the setting
treated in this paper is well defined.
\begin{Them}\label{theorem41} Assume that $\lambda > \mu $
and define process $( R^{\pi_{b_0}^\ast,x}_t, L^{\pi_{b_0}^\ast}_t)$
 by the
following SDE:
\begin{eqnarray}\label{41}
\left\{
\begin{array}{l l l}
dR_t^{\pi_{b_o}^\ast,x}=(\mu-(1-A^*_{b_0}(R^{\pi_{b_o}^\ast,x}_t)
)\lambda)dt+\sigma A^*_{b_0}(R^{\pi_{b_o}^\ast,x}_t)d
{W}_{t}-dL_t^{\pi_{b_o}^\ast},\\
R_0^{\pi_{b_o}^\ast,x}=x,\\
0\leq R^{\pi_{b_o}^\ast, x}_t\leq  b_0,\\
\int^{\infty}_0 I_{\{t: R^{\pi_{b_o}^\ast,x}_t
<b_0\}}(t)dL_t^{\pi_{b_o}^\ast}=0.
\end{array}\right.
\end{eqnarray}
Then for any $0< x\leq b_0$ there exists $\varepsilon_0 >0 $ such
that
\begin{eqnarray}\label{42}
\mathbf{P}\{\tau_{x}^{\pi_{b_0}^\ast} \leq T\}\geq \varepsilon_0>0,
\end{eqnarray}
where $\varepsilon_0= \min\big \{ \frac{
4[1-\Phi(\frac{x}{d\sigma\sqrt{T}})]^2     }{
\exp\{\frac{2}{\sigma^2}( \lambda^2 +\delta^2 )T\} },
\frac{x}{\sqrt{2\pi}\sigma}\int^T_0t^{-\frac{3}{2}}\exp\{-\frac{(x+\mu
t)^2}{2\sigma^2 t}\}dt          \big \}.$
\end{Them}
\begin{proof}  We first consider the case of $\mu < \lambda < 2\mu $.
Denote by  $\delta $ the $\lambda-\mu$ and define new process
$R_t^{(1),x}$ by
\begin{eqnarray}\label{43}
\left\{\begin{array}{l l l} d
R_t^{(1),x}=(\mu-(1-A^*_{b_0}(R_t^{(1),x}) )\lambda)dt+\sigma
A^*_{b_0}(R_t^{(1),x})d {W}_{t},\\
\ R_0^{(1),x}=x.
\end{array}\right.
\end{eqnarray}
By using comparison theorem on SDE(see Ikeda and Watanabe
\cite{s107}and \cite{IW}),
\begin{eqnarray}\label{44}
 \mathbf{P}\{ R_t^{\pi_{b_o}^\ast,x}\leq R_t^{(1),x}
\}=1.
\end{eqnarray}
Define a measure  $\mathbf{Q}$ on $\mathcal{F}_T $ by
\begin{eqnarray}\label{45}
d\mathbf{Q}(\omega)=M_T(\omega)d\mathbf{P}(\omega),
\end{eqnarray}
where
\begin{eqnarray*}
M_{t}&=&\exp\big\{-\int_{0}^{t}\frac{(\lambda
A_{b_0}^*(R_t^{(1),x})-\delta)}{\sigma
A_{b_0}^*(R_t^{(1),x})}dW_{s}- \frac{1}{2}\int_{0}^{t}\frac{(\lambda
A_{b_0}^*(R_t^{(1),x})-\delta)^{2}}{[\sigma
A_{b_0}^*(R_t^{(1),x})]^{2}}ds\big
\}\\
&:=& \mathcal{E}(N)_t=\exp\{N_t-\frac{1}{2}<N>_t\},
\end{eqnarray*}
$N_t= -\int_{0}^{t}\frac{(\lambda
A_{b_0}^*(R_t^{(1),x})-\delta)}{\sigma A_{b_0}^*(R_t^{(1),x})}dW_{s}
$ and $<N>$ is its  bracket. \vskip 10pt \noindent
 By Corollary A.1 in appendix,  $\{M_t\}$ is an exponential martingale
w.r.t.$\mathcal{F}_t$. So by Girsanov theorem, $\mathbf{Q}$ is a
probability measure on $\mathcal {F}_T$ and $ R_t^{(1),x}   $
satisfies  the following SDE:
\begin{eqnarray}\label{46}
dR_t^{(1),x}=A_{b_0}^*(R_t^{(1),x})\sigma \tilde{W}_t,\
R_0^{(1),x}=x,
\end{eqnarray}
where $\tilde{W}_t$ is a standard Brownian motion w.r.t $\bf{Q}$.
\vskip 10pt \noindent By Corollary A.1 in appendix, we can define a
time-change $\rho(t)$ and a processes $ \hat{R}_t^{(1),x} $ by
\begin{eqnarray}\label{47}
\dot{\rho}(t)=\frac{1}{{A_{b_0}^*}^2(R_t^{(1),x})\sigma^2},
\end{eqnarray}
and
$$\hat{R}_t^{(1),x}=R_{\rho(t)}^{(1),x}, $$ respectively. Then $\rho(t)$ is a
strictly increasing  w.r.t. $t$ and (\ref{46}) becomes
\begin{eqnarray*}
\hat{R}_t^{(1),x} =x+\hat{W}_t,
\end{eqnarray*}
where $\hat{W}_t$ is a standard Brownian motion w.r.t $\mathbf{Q}$.
Moreover, by the part (ii) of Corollary A.1 in appendix, we know
that for $t\geq 0$
\begin{eqnarray}\label{48}
\dot{\rho}(t)&=&\frac{1}{{A_{b_0}^*}^2(R_t^{(1),x})\sigma^2}\leq
\frac{1}{d^2\sigma^2}.
\end{eqnarray}
So $\rho(t)\leq \frac{1}{d^2\sigma^2} t$ and $\rho^{-1}(t) \geq
d^2\sigma^2 t$. Therefore
\begin{eqnarray}\label{49}
{\bf{Q}}[\inf\{t: R_t^{(1),x} \leq 0\} \leq T]&=&{\bf{Q}}[\inf\{t:
\hat{R}_{\rho^{-1}(t)}^{(1),x}\leq 0\}\leq
T]\nonumber\\&=&{\bf{Q}}[\inf\{\rho(t): x+\hat{W}_t\leq 0\}\leq T]\nonumber\\
&=&{\bf{Q}}[\inf \{t: \hat{W}_t\leq - x \}\leq
\rho^{-1}(T)]\nonumber\\&\geq& {\bf{Q}}[\inf\{t: \hat{W}_t\leq
-x\}\leq d^2\sigma^2
T]\nonumber\\
&=&2[1-\Phi(\frac{x}{d\sigma\sqrt{T}})]>0,
\end{eqnarray}
where $\Phi(\cdot)$ is the standard normal distribution function. By
(\ref{45}),  we have
\begin{eqnarray}\label{410}
\mathbf{Q}[\inf\{t:R_t^{(1),x}\leq 0\}\leq T]&=&\int {I}_{[\inf\{t:
R_t^{(1),x}\leq 0\}\leq T]}d\mathbf{Q}(\omega)\nonumber\\&=&\int
{I}_{[\inf\{t: R_t^{(1),x}\leq 0\}\leq
T]}M_Td\mathbf{P}(\omega)\nonumber\\&=&\mathbf{E}
[M_T{I}_{[\inf\{t:R_t^{(1),x}\leq 0\}\leq
T]}]\nonumber\\
&\leq&\mathbf{E}[M_T^2]^{\frac{1}{2}}\mathbf{P}[\inf\{t:R_t^{(1),x}\leq
0\}\leq T]^{\frac{1}{2}}.\nonumber\\
\end{eqnarray}
By using Corollary A.1 in appendix again,
\begin{eqnarray}\label{411}
  \mathbf{E}[M_T^2]& =&\mathbf{E} \{\mathcal{E}(2N)_T
\exp\{<N>_T\} \}\nonumber\\&\leq & \exp\{\frac{2}{\sigma^2}(
\lambda^2 +\delta^2
)T\}\mathbf{E}\{\mathcal{E}(2N)_T\}\nonumber\\
&=& \exp\{\frac{2}{\sigma^2}( \lambda^2 +\delta^2 )T\}.
\end{eqnarray}
We deduce from (\ref{49}), (\ref{410}) and (\ref{411}) that
\begin{eqnarray}\label{412}
\mathbf{P}[\inf\{t:R_t^{(1),x}\leq 0\}\leq T]\geq \frac{
4[1-\Phi(\frac{x}{d\sigma\sqrt{T}})]^2     }{
\exp\{\frac{2}{\sigma^2}( \lambda^2 +\delta^2 )T\}  },
\end{eqnarray}
which, together with (\ref{44}) and (\ref{412}), implies that
\begin{eqnarray*}
\mathbf{P}\{\tau_{x}^{\pi_{b_0}^\ast} \leq T    \}\geq \frac{
4[1-\Phi(\frac{x}{d\sigma\sqrt{T}})]^2     }{
\exp\{\frac{2}{\sigma^2}( \lambda^2 +\delta^2 )T\} }.
\end{eqnarray*}
Next we consider the case of $ \lambda \geq 2\mu $. Since
$A_{b_0}^*(x)=1$ for any $x\geq 0$, we have (cf.\cite{s110})
\begin{eqnarray*}\label{413}
\mathbf{P}\{\tau_{x}^{\pi_{b_0}^\ast} \leq T    \} &\geq &\mathbf{P}
\{\inf\{t: \mu t + \sigma W_t= -x \}\leq T      \}\nonumber\\
&=&\frac{x}{\sqrt{2\pi}\sigma}\int^T_0t^{-\frac{3}{2}}\exp\{-\frac{(x+\mu
t)^2}{2\sigma^2 t}\}dt>0.
\end{eqnarray*}
Thus if let  $$\varepsilon_0 = \min\big \{ \frac{
4[1-\Phi(\frac{x}{d\sigma\sqrt{T}})]^2     }{
\exp\{\frac{2}{\sigma^2}( \lambda^2 +\delta^2 )T\} },
\frac{x}{\sqrt{2\pi}\sigma}\int^T_0t^{-\frac{3}{2}}\exp\{-\frac{(x+\mu
t)^2}{2\sigma^2 t}\}dt          \big \},$$ then the proof follows.
\end{proof}
\vskip 10pt \noindent
\begin{Them}\label{lemma51}
 Assume that $\delta=\lambda-\mu>0$
 and define $( R^{\pi_{b}^\ast,b}_t,L_t^{\pi_{b}^\ast})$ by the
following SDE:
\begin{eqnarray}\label{414}
\left\{
\begin{array}{l l l}
dR_t^{\pi_{b}^\ast,b}=(\mu-(1-A^*_{b}(R^{\pi_{b}^\ast,b}_t)
)\lambda)dt+\sigma A^*_{b}(R^{\pi_{b}^\ast,b}_t)d
{W}_{t}-dL_t^{\pi_{b}^\ast},\\
R_0^{\pi_{b}^\ast,b}=b,\\
0\leq R^{\pi_{b}^\ast, b}_t\leq  b,\\
\int^{\infty}_0 I_{\{t: R^{\pi_{b}^\ast,b}_t
<b\}}(t)dL_t^{\pi_{b}^\ast}=0.
\end{array}\right.
\end{eqnarray}
Then
\begin{eqnarray}\label{415}
\lim_{b\rightarrow  \infty}\mathbf{P}[ \tau_{b}^{\pi^*_b}\leq
T]=0,\end{eqnarray} where $ \tau_{b}^{\pi^*_b}=\inf\{t:
R^{\pi_{b}^\ast, b}_t=0\}$.
\end{Them}
\begin{proof} We only need to prove Theorem \ref{lemma51}
in case of $ \mu< \lambda < 2\mu$ because other case can be treated
similarly.\\
For large $b> m $, by the same way as in proving Theorem 3.1 of
\cite{ime02}, we have
\begin{eqnarray*}
{\bf P}\{\tau_{b}^{\pi^*_b}\leq T\}\leq {\bf
P}\{\tau_{(b+m)/2}^{\pi^*_b}\leq T\}.
\end{eqnarray*}
It easily follows that
\begin{eqnarray*}
\mathbf{P}\{\tau_{(b+m)/2}^{\pi^*_b}\leq T\}
&\leq&\mathbf{P}\{R_t^{(1)}=m \ \mbox{or}\ R_t^{(1)}
=b\ \mbox{for some $t\geq 0$ }\}\nonumber\\
&\leq& \mathbf{P}\{\sup_{0\leq t \leq T} R_t^{(1)}\geq
b\}+\mathbf{P}\{\inf_{0\leq t \leq T} R_t^{(1)}\leq m\},
\end{eqnarray*}
where  the process $ \{R_t^{(1)}\}  $ satisfies the following
stochastic differential equation,
\begin{eqnarray}\label{416}
\left\{
\begin{array}{l l l}
dR^{(1)}_{t}=(\mu-(1-A^*_{b}(R^{(1)}_{t}) )\lambda)dt+\sigma
A^*_{b}(R^{(1)}_{t})d{W}_{t},\\
 R^{(1)}_{0}=(b+m)/2.
 \end{array}\right.
\end{eqnarray}
Define measure  $\mathbb{Q}_1$ on $\mathcal{F}_T $ by
\begin{eqnarray}\label{417}
d\mathbf{P}(\omega)=\widetilde{M}_T(\omega)d \mathbb{Q}_1(\omega),
\end{eqnarray}
where
\begin{eqnarray*}
\widetilde{M}_{t}&=&\exp\big\{\int_{0}^{t}\frac{(\lambda
A_{b}^*(R_t^{(1)})-\delta)}{\sigma A_{b}^*(R_t^{(1)})}dW_{s}+
\frac{1}{2}\int_{0}^{t}\frac{(\lambda
A_{b}^*(R_t^{(1)})-\delta)^{2}}{[\sigma
A_{b}^*(R_t^{(1)})]^{2}}ds\big \}.
\end{eqnarray*}
By Corollary A.1 in appendix, $\{\widetilde{M}_{t}\}$ is an
exponential martingale. So by Girsanov theorem, $\mathbb{Q}_1$ is a
probability measure on $\mathcal {F}_{T}$ and the process
$\hat{\mathcal{W}_{t}}:=\int_{0}^{t}\frac{(\lambda
A_{b}^*(R_s^{(1)})-\delta)}{\sigma A_{b}^*(R_s^{(1)})}ds+{W}_{t}$ is
a Brownian motion w.r.t.$\mathbb{Q}_1$, as well as the
SDE(\ref{416}) becomes
$$dR_t^{(1)}=\sigma
A^*_{b}(R^{(1)}_{t})\hat{\mathcal{W}_{t}}, \ R^{(1)}_{0}=(b+m)/2,\
\mbox{a.e.}, \  \mathbb{Q}_1
$$
\vskip 10pt \noindent
 Firstly, we estimate the  term  $
\mathbf{P}\{\sup_{0\leq t \leq T} R_t^{(1)}\geq b\}$. \vskip 10pt
\noindent By (\ref{417}), H\"{o}lder's inequality, Chebyshev
inequality and B-D-G inequalities (see Ikeda and Watanabe
\cite{s107}(1981)),  we have
\begin{eqnarray}\label{418}
 \mathbf{P}\{\sup_{0\leq t \leq
T} R_t^{(1)}\geq b\}&\leq & [{\bf E}^{\mathbb{Q}_1}\{
\widetilde{M}^2_T \}]^{\frac{1}{2}} \mathbb{Q}_1\{ \sup_{0\leq t
\leq T} R_t^{(1)}\geq b    \}^{\frac{1}{2}},
\end{eqnarray}
\begin{eqnarray}  \label{419}
\mathbb{Q}_1\{\sup_{0\leq t\leq T}R_t^{(1)}\geq b \} &\leq
&{\mathbb{Q}_1}\{\sup_{0\leq t\leq T}|\int_0^t \sigma A^*_{b}
(R^{(1)}_{s})d\hat{\mathcal{W}_{s}} |\geq \frac{b-m}{2}\}\nonumber\\
&\leq &\frac{ 4{\bf E}^{\mathbb{Q}_1} \{ \sup_{0\leq t \leq
T}|\int_0^t \sigma A^*_{b}(R^{(1)}_{s})
d\hat{\mathcal{W}_{s}}|    \}^2}{(b-m)^2 }\nonumber\\
&\leq  & \frac{ 16{\bf E}^{\mathbb{Q}_1} \{ \int_0^T ( \sigma
A^*_{b} (R^{(1)}_{s}))^2
 ds|    \}}{(b-m)^2 }\nonumber\\
 &\leq & \frac{16T\sigma^2 \widetilde{B}^2}{(b-m)^2},
\end{eqnarray}
where ${\bf E}^{\mathbb{Q}_1}$ denotes mathematical expectation with
 respect to the probability measure $\mathbb{Q}_1$.
\vskip 10pt \noindent Secondly, we estimate the term $P\{\inf_{0\leq
t \leq T} R_t^{(1)}\leq m\}       $ as follows. \vskip 10pt
\noindent Noting that $A^*_{b} (x)=1$ for $x\geq m$, we have
\begin{eqnarray}\label{420}
\mathbf{P}\{\inf_{0\leq t \leq T} R_t^{(1)}\leq
m\}&=&1-\mathbf{P}\{\inf_{0\leq t \leq
T} R_t^{(1)}> m\}\nonumber \\
&=&1- \mathbf{P}\{\inf_{0\leq t\leq T}\{\mu t +\sigma
W_t\}>-\frac{b-m}{2}
\}\nonumber\\
&&\rightarrow 1-1=0 \ \mbox{as} \ b\rightarrow \infty.
\end{eqnarray}
By the same way as in the proof of (\ref{411}),
\begin{eqnarray}\label{421}
{\bf E}^{\mathbb{Q}_1}\{ \widetilde{M}^2_T \}\leq C(T)< \infty.
\end{eqnarray}
Therefore, the equality (\ref{415}) easily follows from the
inequalities (\ref{418})-(\ref{421}).
 \end{proof}
 \vskip 15pt \noindent
\setcounter{equation}{0}
\section{\bf Numerical analysis}
\vskip 5pt \noindent In this section we present  numerical results
to demonstrate how the volatility $\sigma^2$, the premium rate $\mu
$ and the initial capital $x$ impact on the company's safety and
profit and how the risk $\varepsilon$ effect on risk-based capital
standard $x$, optimal retention ratio, optimal dividend payout
level, optimal control policy and the company's profit. Inversely,
we also explain how the risk $\varepsilon$ impacts on optimal
dividend payout level  based on PDE (\ref{422}), the probability of
bankruptcy and value function below. \noindent
\begin{example}
The graphs \ref{difsigma} and \ref{difmu} below show that the value
$g(x)$ increases with $(x, \mu, \sigma^2 )$, so higher the
volatility $\sigma^2$, the premium rate $\mu $ and the initial
capital $x$ will make the company get  more return.\noindent
\begin{figure}[H]
\resizebox{7cm}{7cm}{\includegraphics{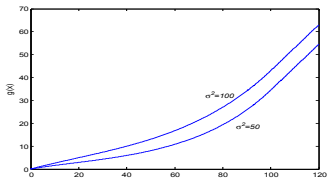}} \caption{Value
$g(x)$ as a function of $ x$ in cases $\sigma^2=50 $ and $
\sigma^2=100$, respectively (parameters: $\mu=2$; $\lambda=6$;
$c=0.05$; $b=100$)}\label{difsigma}
\end{figure}
\begin{figure}[H]
\resizebox{7cm}{7cm}{\includegraphics{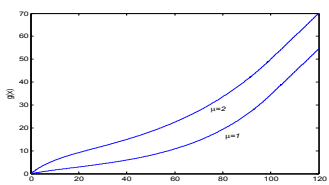}} \caption{ Value
$g(x)$ as a function of $ x$ in cases $\mu =1$ and $\mu =2 $,
respectively (parameters: $\sigma^2=50$; $\lambda=6 $; $c=0.05 $;
$b=100$)}\label{difmu}
\end{figure}
\end{example}
\begin{example} The graph \ref{psi-x} below shows
that the probability of bankruptcy $ \psi=1-\phi$ decreases with the
initial capital $x$.
\begin{figure}[H]
\resizebox{7cm}{7cm}{\includegraphics{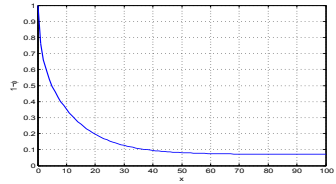}} \caption{ The
probability of bankruptcy $ \psi=1-\phi$ as a function of $x$
(parameters: $\sigma^2=50$; $\mu=2 $; $\lambda=0.4$; $c=0.05$;
$b=100$; $T=500$ )}\label{psi-x}
\end{figure}
\end{example}
\begin{example}
The graph \ref{ep-x} below shows that the risk-based capital
standard  $x$  decreases with the preferred risk level
$\varepsilon$. It states that the higher preferred risk level
$\varepsilon$ needs a lower the initial risk-based capital.
\begin{figure}[H]
\resizebox{7cm}{7cm}{\includegraphics{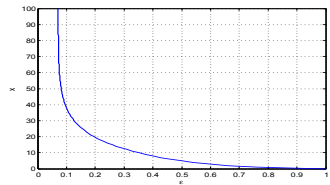}} \caption{ The
risk-based capital standard  $x$ as a function of $
\varepsilon$(parameters: $ \sigma^2=50$; $\mu=2$; $\lambda=0.4$;
$c=0.05$; $b=100$; $T=500$)}\label{ep-x}
\end{figure}
\end{example}
\begin{example}
The graphs \ref{xinep-b} and \ref{xinb-ep} below show that the
optimal dividend level $ b(\varepsilon )$ is a decreasing function
of the risk level $ \varepsilon $. Inversely, the risk level
$\varepsilon(b) $ is also
 a decreasing function of $b$.
\begin{figure}[H]
\resizebox{7cm}{7cm}{\includegraphics{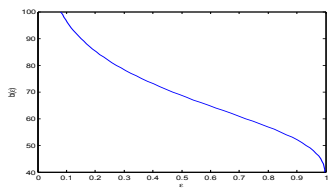}}
 \caption{ The optimal dividend level $ b(\varepsilon )$ as a function of
$\varepsilon$ (parameters: $\sigma^2=50$; $\mu=2$; $\lambda=0.4$;
$c=0.05 $; $b=100$; $T=500$ )}\label{xinep-b}
\end{figure}
\begin{figure}[H]
\resizebox{7cm}{7cm}{\includegraphics{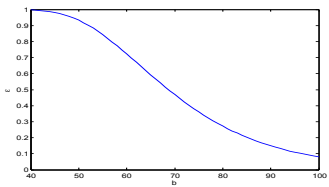}}
 \caption{ The risk level $ \varepsilon(b)$ as a function of $b$ (parameters:
$\sigma^2=50$; $\mu=2$; $\lambda=0.4$; $c=0.05$; $b=100$; $T=500$
)}\label{xinb-ep}
\end{figure}
\end{example}
\vskip 10pt \noindent
 \setcounter{equation}{0}
\section{{\small {\bf Properties on the probability of bankruptcy }} }
\vskip 10pt \noindent To give the proof of main result (Theorem
\ref{theorem33}) of this paper we  list some lemmas on properties of
the probability of bankruptcy in this section, and their  detailed
proofs will be given in section 8. \vskip 5pt \noindent
\begin{lemma}\label{lemma52} Assume that $\delta=\lambda-\mu>0$,
 $R^{\pi_{b}^\ast,b}_t$
and $\tau_{b}^{\pi_b^*}$ are the same as in Lemma \ref{lemma51}.
Then the probability of bankruptcy $\mathbf{P}[\tau_b^b\leq T]$ is
strictly decreasing on $ [m, b_K]$,where
$\tau_b^b:=\tau_{b}^{\pi_b^*}$ and $b_K:=\inf\{b:
\mathbf{P}[\tau_b^b\leq T]=0\}$.
\end{lemma}
\vskip 5pt \noindent
\begin{Prop}\label{proposition51}
 Assume that $\delta=\lambda-\mu>0$
 and  for any $b\geq b_0 $ we define $R^{\pi_{b}^\ast,y}_t$ by
 (\ref{41}) with initial value $y$.
 Let $\phi^{b}(t,y)\in C^1(0,\infty)\cap C^2(0,b)$ and satisfy the
following partial differential equation with  boundary conditions,
\begin{eqnarray}\label{422}
\left\{
\begin{array}{l l l}
\phi_{t}^{b}(t,y)=\frac{1}{2}[A_b^{*}(y)]^2\sigma^2\phi_{yy}^{b}(
t,y)+(\lambda A_b^{*}(y)-\delta)\phi_{y}^{b}(t,y),\\
\phi^{b}(0,y)=1,\  \mbox{for}\ \  0<y\leq b, \\
\phi^{b}(t,0)=0,\phi_{y}^{b}(t,b)=0,\ \mbox{for} \ t>0.
\end{array}\right.
\end{eqnarray}
Then $\phi^{b}(T,y)=1-\psi^{b}(T,y)$, i.e., $\phi^{b}(T,y)$ is
probability that the company will survive on $[0, T]$, where
$\psi^{b}(T,y)= \mathbf{P}\{\tau^{\pi_{b}^\ast}_y\leq T\}.$
\end{Prop}
\begin{Remark}\label{R41} If we define $a(y):=\frac{1}{2}[A_b^{*}(y)]^2\sigma^2$,
$\mu(y):=\lambda A_b^{*}(y)-\delta$, then the equation (\ref{422})
becomes
\begin{eqnarray*}
\phi_{t}^{b}(t,y)=a(y)\phi_{yy}^{b}( t,y)+\mu(y)\phi_{y}^{b}(t,y).
\end{eqnarray*}
By the properties of $A_b^{*}(y)$, we can easily show that $a(y)$
and $\mu(y)$ are continuous in $[0,b]$. So there exists a unique
solution in $C^1(0,\infty)\cap C^2(0,b)$ for (\ref{422}). Moreover,
by Corollary A.\ref{corollary31} in appendix, $a^{'}(y)$,
$\mu^{'}(y)$ and $a^{''}(y)$ are bounded in $(0,m)$ and $(m,b)$.
\end{Remark}
\vskip 10pt \noindent The following,  together with Proposition
\ref{proposition51}, states that
 the  probability of bankruptcy
$\psi^{b}(T, b):={\bf P}\big \{ \tau_{b}^{\pi^*_b}< T \big\} $ is
continuous  with respect to $b(b\geq b_0)$.
 \vskip 10pt \noindent
\begin{lemma}\label{lemma53}
Assume that the same conditions as in Lemma \ref{lemma52}. Let the
function $\phi^{b}(t,x)$ solve the equation(\ref{422}) and
$u(b):=\phi^{b}(T,b)$. Then $u(b)$ is continuous with respect to
$b(b\geq b_0 )$.
\end{lemma}

\vskip 10pt \noindent
 As a direct consequence of Theorem
\ref{lemma51}, Lemmas \ref{lemma52} and \ref{lemma53} we have the
following. \vskip 5pt \noindent
\begin{Them}\label{theorem51} For any fixed $\varepsilon >0$,
there exists a unique $ b^*$ satisfying ${\bf P}\big \{
\tau_{b^*}^{\pi^*_{b^*}}< T \big\}=\varepsilon$.
\end{Them}
\setcounter{equation}{0}
\section{ {\bf Proof of Main Results}}
 \vskip 15pt\noindent
In this section we give the  proof of the main results of this
paper. we  first need the following.
\begin{Them}\label{theorem61}
Let $V(x,b)$ and $V(x)$ be defined by (\ref{f3}) and (\ref{f4}),
$f(x)$, $g(x)$ and $A^*_b(x)$ be the same as in  Lemma A.
\ref{theorem31} and Lemma A.\ref{theorem32} in appendix,
respectively. We have the following. \vskip 10pt \noindent
 (i) If $b\leq b_0$, then
$V(x,b)=V(x,b_0)=J(x,\pi_{b_o}^\ast)=V(x,0)=f(x)$ and the optimal
policy $\pi_{b_o}^\ast=\{
 A^*_{b_0}(R^{\pi_{b_o}^\ast}_\cdot),L^{\pi_{b_o}^\ast}_\cdot\}$ is uniquely
determined by the following SDE:
\begin{eqnarray}\label{342}
\left\{
\begin{array}{l l l}
dR_t^{\pi_{b_o}^\ast}=(\mu-(1-A^*_{b_0}(R^{\pi_{b_o}^\ast}_t)
)\lambda)dt+\sigma A^*_{b_0}(R^{\pi_{b_o}^\ast}_t)d
{W}_{t}-dL_t^{\pi_{b_o}^\ast},\\
R_0^{\pi_{b_o}^\ast}=x,\\
0\leq R^{\pi_{b_o}^\ast}_t\leq  b_0,\\
\int^{\infty}_0 I_{\{t: R^{\pi_{b_o}^\ast}_t
<b_0\}}(t)dL_t^{\pi_{b_o}^\ast}=0.
\end{array}\right.
\end{eqnarray}
(ii) If $b\geq b_0$, then $V(x,b)=J(x,\pi^*_{b})=  g(x)$ and the
optimal policy $\pi_{b}^\ast=\{
 A^*_b(R^{\pi_{b}^\ast}_\cdot),L^{\pi_{b}^\ast}_\cdot\}$
is uniquely determined by the following SDE:
\begin{eqnarray}\label{343}
\left\{
\begin{array}{l l l}
dR_t^{\pi_{b}^\ast}=(\mu-(1-A^*_b(R^{\pi_{b}^\ast}_t)
)\lambda)dt+\sigma A^*_b(R^{\pi_{b}^\ast}_t)d
{W}_{t}-dL_t^{\pi_{b}^\ast},\\
R_0^{\pi_{b}^\ast}=x,\\
0\leq R^{\pi_{b}^\ast}_t\leq  b,\\
\int^{\infty}_0 I_{\{t: R^{\pi_{b}^\ast}_t
<b\}}(t)dL_t^{\pi_{b}^\ast}=0.
\end{array}\right.
\end{eqnarray}
\end{Them}
\begin{proof}
(i) Assume that $b\leq b_0$. By using the fourth equality in
(\ref{342}), it follows that $\pi^*_{b_0}\in \Pi_{b_0}\subset\Pi_b
$, so $V(x, b_0)\leq V(x,b)\leq V(x,0)$. Therefore it suffices to
prove that $V(x,0)\leq f(x)=V(x, b_0)=J(x,\pi_{b_o}^\ast)   $. For
any admissible policy $\pi=\{a_\pi, L^\pi\}$, we assume that the
process $( R^\pi_t, L^\pi_t)$ satisfies (\ref{21}). Let
$\Lambda=\{s:L_{s-}^{\pi}\neq L_{s}^{\pi}\}$ and
$\hat{L}=\sum_{s\in\Lambda, s\leq t}(L_{s}^{\pi}-L_{s-}^{\pi})$ be
the discontinuous part of $L_{s}^{\pi}  $ and
$\tilde{L}_{t}^{\pi}=L_{t}^{\pi}-\hat{L}_{t}^{\pi}$ be the
continuous part of $L_{s}^{\pi}$, respectively. Define
$\tau^\varepsilon =\inf\{t\geq 0: R^{\pi}_{t}\leq \varepsilon\}$. By
applying generalized It\^{o} formula to stochastic process $ R^\pi_t
$ and the function $f(x)$, we get
\begin{eqnarray}\label{344}
e^{-c(t\wedge \tau^{\varepsilon})}f(R_{t\wedge
\tau^{\varepsilon}}^{\pi})&=
&f(x)+\int_{0}^{t\wedge\tau^{\varepsilon}}e^{-cs}\mathcal
{L} f(R_{s}^{\pi})ds\nonumber\\
&+&\int_{0}^{t\wedge\tau^{\varepsilon}}a_\pi\sigma
e^{-cs}f^{'}(R_{s}^{\pi})d{W}_{s}\mathcal
-\int_{0}^{t\wedge\tau^{\varepsilon}}e^{-cs}f^{'}(R_{s}^{\pi})dL_{s}^{\pi}\nonumber\\
&+&\sum\limits_{s\in\Lambda ,s\leq t\wedge
\tau^{\varepsilon}}e^{-cs}[f(R_{s}^{\pi})-f(R_{s-}^{\pi})\nonumber\\
&-&f^{'}(R_{s-}^{\pi})(R_{s}^{\pi}-R_{s-}^{\pi})]\nonumber \\
&= &f(x)+\int_{0}^{t\wedge\tau^{\varepsilon}}e^{-cs}\mathcal
{L}  f(R_{s}^{\pi})ds\nonumber\\
&+&\int_{0}^{t\wedge\tau^{\varepsilon}}a_\pi \sigma
e^{-cs}f^{'}(R_{s}^{\pi})d {W}_{s}\mathcal
-\int_{0}^{t\wedge\tau^{\varepsilon}}e^{-cs}f^{'}(R_{s}^{\pi})d\tilde{L}_{s}^{\pi}\nonumber\\
&+&\sum\limits_{s\in\Lambda ,s\leq t\wedge
\tau^{\varepsilon}}e^{-cs}[f(R_{s}^{\pi})-f(R_{s-}^{\pi}))],
\end{eqnarray}
where
\begin{eqnarray*}
\mathcal
{L}=\frac{1}{2}a^{2}\sigma^{2}\frac{d^{2}}{dx^{2}}+(\mu-(1-a)\lambda)\frac{d}{dx}-c.
\end{eqnarray*}
By  the (\ref{34}) and $f^{'}(R_{s\wedge
\tau^\varepsilon}^{\pi})\leq f^{'}(\varepsilon)$, the second term
and third term on the right hand side of (\ref{344}) is non-positive
and  a square integrable martingale, respectively, therefore, by
taking mathematical expectations at both sides of (\ref{344}) and
letting $\varepsilon\rightarrow 0$, we have
\begin{eqnarray}\label{345}
{\bf E}\big \{e^{-c(t\wedge \tau ^\pi _x)}f(R_{t\wedge \tau ^\pi
_x}^{\pi})\big \}&\leq& f(x)-{\bf E}\big \{\int_{0}^{t\wedge \tau
^\pi _x}e^{-cs}
f^{'}(R_{s}^{\pi})d\tilde{L}_{s}^{\pi}\big \}\nonumber\\
&+&{\bf E}\big\{\sum\limits_{s\in\Lambda ,s\leq t\wedge
\tau ^\pi _x}e^{-cs}[f(R_{s}^{\pi})-f(R_{s-}^{\pi})]\big \}.\nonumber\\
\end{eqnarray}
Since $f^{'}(R_{s}^{\pi})\geq 1$ for $x\geq 0$,
\begin{eqnarray}\label{346}
f(R_{s}^{\pi})-f(R_{s-}^{\pi})\leq-(L_{s}^{\pi}-L_{s-}^{\pi}),
\end{eqnarray}
 which, together with (\ref{345}), implies
that
\begin{eqnarray}\label{347}
{\bf E}\big\{e^{-c(t\wedge \tau ^\pi _x)}f(R_{t\wedge \tau ^\pi
_x}^{\pi})\big\}&+&{\bf E}\big\{\int_{0}^{ t\wedge \tau ^\pi
_x}e^{-cs}dL_{s}^{\pi}\big \}\leq f(x).
\end{eqnarray}
By definition of $\tau ^\pi _x $ and $f(0)=0$, it is easy to prove
that
\begin{eqnarray}\label{348}
\liminf\limits_{t\rightarrow\infty}e^{-c(t\wedge \tau ^\pi
_x)}f(R_{t\wedge
\tau ^\pi _x}^{\pi})&=&e^{-c\tau}f(0)I_{\{\tau ^\pi _x<\infty\}}\nonumber\\
&+& \liminf\limits_{t\rightarrow\infty}e^{-ct}f(R_{t})I_{\{\tau ^\pi
_x =\infty\}}\geq 0.
\end{eqnarray}
So we see from (\ref{348}) and (\ref{347}) that
\begin{eqnarray*}
J(x,\pi)={\bf E}[\big\{\int_{0}^{\tau ^\pi _x} e^{-cs}dL_{s}^{\pi}
\}]\leq f(x).
\end{eqnarray*}
Thus
\begin{eqnarray}
V(x,0)\leq f(x).
\end{eqnarray}
If we let policy $\pi_{b_o}^\ast=\{
 A^*_{b_0}(R^{\pi_{b_o}^\ast}_\cdot),L^{\pi_{b_o}^\ast}_\cdot\}$ ,
  which is uniquely determined by SDE(\ref{342}),
see Lions and Sznitman\cite{s10100}, then $R^{\pi_{b_o}^\ast}_t$ and
$ L^{\pi_{b_o}^\ast}_t$ are continuous  stochastic processes. So all
the inequalities above become equalities and
\begin{eqnarray*}
V(x,0)\leq f(x)=V(x,b_0)= J(x,\pi_{b_o}^\ast)    .
\end{eqnarray*}
The proof of the part (i) follows.

(b) We assume that $b\geq b_0$. For any $\pi\in \Pi_b $, let  $(
R^\pi_t, L^\pi_t)$ satisfies  (\ref{21}). It is easy to see from
the definition of $\Pi_b$ that
\begin{eqnarray}\label{flemma3}
\left\{
\begin{array}{l l l}
\mathbf{P}\{ R_{s-}^\pi\geq R_{s}^\pi\geq b\}+\mathbf{P}\{b\geq
R_{s-}^\pi\geq R_{s}^\pi\} =1,
\ \forall s\geq 0 ,\\
\mathcal {L}g(R_s^\pi)\leq 0 \  \mbox{for $ s\leq \tau ^\pi _x
=\inf\{t\geq 0: R^\pi_s\leq 0\} $},\\
g'(x)=1 \ \mbox{for $x\geq b$}.
\end{array}\right.
\end{eqnarray}
By using (\ref{flemma3}), we have (\ref{346}) with replacing $f$ by
$g$. Then by the same way as in (i),
$$ V(x,b)\leq g(x).$$
Choosing the policy $\pi_{b}^\ast=\{
 A^*_b(R^{\pi_{b}^\ast}_\cdot),L^{\pi_{b}^\ast}_\cdot\}$,
 which is uniquely determined by SDE(\ref{343}),
yields that the last inequality becomes equality. Thus the proof is
complete.
\end{proof}
\vskip 10pt \noindent Now we give the proof of main result (Theorem
\ref{theorem33}) of this paper. \vskip 10pt \noindent
\begin{proof}
If $\mathbf{P}[\tau_{b_0}^{\pi^*_{b_0}}\leq T]\leq\varepsilon$, then
the conclusion follows from the part (i) of Theorem \ref{theorem61}.
\vskip 7pt\noindent If  $ \mathbf{P}[\tau_{b_0}^{\pi^*_{b_0}}\leq
T]>\varepsilon $, then, by  Theorem \ref{lemma51}, Lemmas
\ref{lemma52} and \ref{lemma53}, the equation
$\mathbf{P}\{\tau_{b}^{\pi^*_{b}}\leq T\}=\varepsilon $ has a unique
solution $b^\ast$ and $b^\ast=\inf\{ b:  b\in \mathfrak{B} \}      >
b_0$. By Theorem \ref{theorem61} and Corollary A.2 in appendix,
$g(x, b)=V(x,b)$ is decreasing w.r.t.$ b(\geq b_0)$, so (\ref{313})
follows from the part (ii) of Theorem \ref{theorem61}. Moreover,
$g(x, b^*)(=V(x, b^*))$ is the value function of the company, the
optimal policy associated with $g(x)= g(x, b^*)$ is
 $\pi_{b^*}^\ast=\{
 A^*_{b^*}(R^{\pi_{b^*}^\ast}_\cdot),L^{\pi_{b^*}^\ast}_\cdot\}$
which is uniquely determined by SDE(\ref{343}). The inequality
(\ref{316}) is a direct consequence of Corollary A.2 in appendix.
Thus we complete the proof.
\end{proof}
\vskip 10pt \noindent
 \setcounter{equation}{0}
\section{ {\bf Proof of Lemmas and Proposition}}
 \vskip 10pt\noindent
 ${\sl Proof of Lemma \ \ref{lemma52}}$.
We only prove  Lemma \ref{lemma52} in case of $\mu < \lambda < 2\mu
$ because other case can be treated similarly. We prove that the
probability of bankruptcy is strictly decreasing on $[m ,b_K]$, that
is,
\begin{eqnarray*}
\mathbf{P}[\tau_{b_1}^{b_1}\leq T]-\mathbf{P}[\tau_{b_2}^{b_2}\leq
T]>0
\end{eqnarray*}
 for any $b_2>b_1\geq m $.
By comparison theorem,
$$ \mathbf{P}[\tau_{b_1}^{b_1}\leq T]-\mathbf{P}[\tau_{b_2}^{b_2}\leq T]\geq
\mathbf{P}[\tau_{b_1}^{b_2}\leq T]-\mathbf{P}[\tau_{b_2}^{b_2}\leq
T].
$$
The proof can be reduced to proving that
\begin{eqnarray}\label{59}
\mathbf{P}[\tau_{b_1}^{b_2}\leq T]-\mathbf{P}[\tau_{b_2}^{b_2}\leq
T]>0.
\end{eqnarray}
\vskip 10pt \noindent To prove the inequality (\ref{59}) we define
stochastic processes $R_t^{[1]}$ and $R_t^{[2]}$ by the following
SDEs:
\begin{eqnarray*}
dR_t^{[1]}=[A_{b_2}^*(R_t^{[1]})\lambda-\delta]dt+A_{b_2}^*(R_t^{[1]})\sigma
d\mathcal {W}_t-dL_t^{b_2},R_0^{[1]}=b_1,
\end{eqnarray*}
\begin{eqnarray*}
dR_t^{[2]}=[A_{b_2}^*(R_t^{[2]})\lambda-\delta]dt+A_{b_2}^*(R_t^{[2]})\sigma
d\mathcal{W}_t-dL_t^{b_2},R_0^{[2]}=b_2,
\end{eqnarray*}
respectively.\vskip 10pt \noindent Let
$\tau^{b_1}=\inf\limits_{t\geq 0}\{t:R_t^{[2]}=b_1\}$, $A= \{
\tau^{b_1}\leq T\}$ and $B=\big\{\sl R_{t}^{[2]}$ will go to
bankruptcy in a time interval $[\tau^{b_1},\tau^{b_1}+T]$ and
$\tau^{b_1}\leq T \big \}$. Then $\{\tau_{b_2}^{b_2}\leq T\}\subset
B$. Moreover, by using strong Markov property of $R_t^{[2]}$, we
have
\begin{eqnarray*}
\mathbf{P}[\tau_{b_1}^{b_2}\leq T]=\mathbf{P}[B|A].
\end{eqnarray*}
So
\begin{eqnarray*}
 \mathbf{P}[\tau_{b_1}^{b_2}\leq T]-\mathbf{P}[\tau_{b_2}^{b_2}\leq T]
 &\geq &
\mathbf{P}[\tau_{b_1}^{b_2}\leq T]-\mathbf{P}(B)\nonumber\\
&=& \mathbf{P}[\tau_{b_1}^{b_2}\leq T]-\mathbf{P}(A)\mathbf{P}(B|A)\nonumber\\
&=&\mathbf{P}[\tau_{b_1}^{b_2}\leq T](1-\mathbf{P}(A))\nonumber\\
&=&\mathbf{P}[\tau_{b_1}^{b_2}\leq T]\mathbf{P}(A^c).
\end{eqnarray*}
By Theorem \ref{theorem41}, $\mathbf{P}[\tau_{b_1}^{b_2}\leq T]>0$.
Hence we only need to prove  $\mathbf{P}(A^c)>0$.  For doing this we
define stochastic processes $R_t^{[3]}$ and $R_t^{[4]}$ by the
following SDEs:
\begin{eqnarray*}
\left\{
\begin{array}{l l l}
dR_t^{[3]}=[A_{b_2}^*(R_t^{[3]})\lambda-\delta]dt+A_{b_2}^*(R_t^{[3]})\sigma
d\mathcal{W}_t-dL_t^{b_2},\\
R_0^{[3]}=\frac{b_1+b_2}{2},
 \end{array}\right.
\end{eqnarray*}
\begin{eqnarray*}
\left\{
\begin{array}{l l l}
dR_t^{[4]}=[A_{b_2}^*(R_t^{[4]})\lambda-\delta]dt+A_{b_2}^*(R_t^{[4]})\sigma
d{W}_t,\\
R_0^{[4]}=\frac{b_1+b_2}{2}.
 \end{array}\right.
\end{eqnarray*}
Setting $D=\{\inf\limits_{0\leq t\leq T}R_t^{[3]}>b_1\}$ and
    $E=\{\inf\limits_{0\leq t\leq T}R_t^{[4]}>b_1,\sup\limits_{0\leq
t\leq T}R_t^{[4]} <b_2\}$, by comparison theorem on SDE, we have
$\mathbf{P}(A^c)\geq \mathbf{P}(D)\geq \mathbf{P}(E)$. Since
$A_{b_2}^*(x)=1$ for any $x>m$,
\begin{eqnarray}\label{510}
R_t^{[4]}= \frac{b_1+b_2}{2}+ [\lambda-\delta]t+\sigma {W}_t\ \mbox{
on $ E$}.
\end{eqnarray}
We deduce from (\ref{510}) and properties of Brownian motion with
drift  (cf. Borodin and Salminen \cite{s110} (2002)) that
\begin{eqnarray*}
\mathbf{P}(E)=\frac{e^{-\mu'^2T/2}}{\sqrt{2\pi
T}}\sum_{k=-\infty}^{\infty}\int_{b_1/\sigma}^{b_2/\sigma}e^{\mu'(z-x)}
[(e^{-(z-x+\frac{2k(b_2-b_1)}{\sigma})^2/2T})\\
-(e^{-(z+x-\frac{2b_1-2k(b_2-b_1)}{\sigma})^2/2T})]dz>0,\nonumber
\end{eqnarray*}
where $\mu'=(\lambda-\delta)/\sigma$ and
$x=\frac{b_1+b_2}{2\sigma}$. Thus the proof follows. \qquad $\Box$
\vskip 15pt\noindent ${\sl Proof of Proposition\
\ref{proposition51}}$.\qquad Let $\phi(t,y)\equiv \phi^{b}(t,y)$.
Since the stochastic process $( R^{\pi_{b}^\ast,y}_t,
L_t^{\pi_{b}^\ast} )$ is continuous, by applying  the generalized
It\^{o} formula to $(R^{\pi_{b}^\ast,y}_t, L_t^{\pi_{b}^\ast} )$ and
$\phi(t,y)$, we have for $0<y\leq b$
\begin{eqnarray}\label{423}
\phi(T-(t\wedge\tau_{y}^{b}),R^{\pi_{b}^\ast,y}_{t\wedge\tau_{y}^{b}})
&=&\phi(T,y)\nonumber\\
&+&\int_{0}^{t\wedge\tau_{y}^{b}}(\frac{1}{2}A_b^{*2}(R^{\pi_{b}^\ast,y}_{s})
\sigma^{2}\phi_{yy}(T-s,R^{\pi_{b}^\ast,y}_{s})\nonumber\\
&+&(\lambda A^*_b(R^{\pi_{b}^\ast,y}_{s})-\delta)
\phi_{y}(T-s,R^{\pi_{b}^\ast,y}_{s})\nonumber\\
&-&\phi_{t}(T-s,
R^{\pi_{b}^\ast,y}_{s}))ds-\int_{0}^{t\wedge\tau_{y}^{b}}
\phi_{y}(T-s,R^{\pi_{b}^\ast,y}_{s})dL_{s}^{b}\nonumber\\
&+&\int_{0}^{t\wedge\tau_{y}^{b}}
A_b^*(R^{\pi_{b}^\ast,y}_{s})\sigma\phi_{y}(T-s,R^{\pi_{b}^\ast,y}_{s})dW_{s},
\end{eqnarray}
where $ \tau_{y}^{b}\equiv\tau^{\pi_{b}^\ast}_y= \inf\{t: R^{\pi_{b}^\ast,y}_t =0\}    $.\\
 Letting $t=T$ and taking mathematical expectation at both
sides of (\ref{423}) yield that
\begin{eqnarray*}
\phi(T,y)&=&\mathbf{E}[\phi(T-(T\wedge\tau_{y}^{b}),
R^{\pi_{b}^\ast,y}_{T\wedge\tau_{y}^{b}}
)]\nonumber\\
&=&\mathbf{E}[\phi(0, R^{\pi_{b}^\ast,y}_{T})1_{T<\tau_{y}^{b}}]+
\mathbf{E}[\phi(T-\tau_{y}^{b},0)1_{T\geq\tau_{y}^{b}})]\nonumber\\
&=&\mathbf{E}[1_{T<\tau_{y}^{b}}]=1-\psi(T,y).
\end{eqnarray*}\qquad
$\Box$ \vskip 15pt\noindent
 Finally, we will  use PDE method to prove that the probability of bankruptcy
$\psi^{b}(T, b)\:={\bf P}\big \{ \tau_{b}^{\pi^*_b}< T \big\} $ is
continuous w.r.t. $b$. \\
\vskip 15pt\noindent
 ${\sl Proof of Lemma \ \ref{lemma53}}$.\quad
 It suffices to prove that $\phi^{b}(t,x)$ is continuous in $ b $.
Let  $y=bz$ and  $\theta^{b}(t,z)=\phi^{b}(t,by)$, the equation
(\ref{422}) becomes
\begin{eqnarray*}
\left\{
\begin{array}{l l l}
\theta_t^{b}(t,z)=[a(bz)/b^2]\theta_{zz}^{b}(t,z)+[\mu(bz)/b]\theta_z^{b}(t,z),\\
\theta^{b}(0,z)=1, \ \mbox{for}\  0<z\leq 1, \\
\theta^{b}(t,0)=0,\theta_{z}^{b}(t,1)=0,\ \mbox{for} \  t>0.
\end{array}\right.
\end{eqnarray*}
So the proof of Lemma \ref{lemma53} reduces to proving
$\lim\limits_{b_2\rightarrow
b_1}\theta^{b_2}(t,z)=\theta^{b_1}(t,z)$ for fixed $b_1>b_0$.
Setting $w(t,z)=\theta^{b_2}(t,z)-\theta^{b_1}(t,z)$, we  have
\begin{eqnarray}\label{424}
\left\{
\begin{array}{l l l}
w_t(t,z)&=&[a(b_2z)/b_2^2]w_{zz}(t,z)+[\mu(b_2z)/b_2]w_{z}(t,z)\\
&+& \{a(b_2z)/b_2^2-a(b_1z)/b_1^2\}\theta_{zz}^{b_1}(t,z)\\
&+&\{a(b_2z)/b_2^2-a(b_1z)/b_1^2\}\theta_z^{b_1}(t,z),\\
w(0,z)&=&0,\ \mbox{for}\ 0<z\leq 1,\\
w(t,0)&=&0,\ w_x(t,1)=0,\ \mbox{for}\ t>0.
\end{array}\right.
\end{eqnarray}
By multiplying the first equation in (\ref{424}) by $w(t,z)$ and
then integrating  on $ [0,t]\times [0,1]$,
\begin{eqnarray}\label{425}
&&\int_0^{t}\int_0^1 w(s,x)w_t(s,x)dxds
\nonumber\\
&=&\int_0^{t}\int_0^1 \big\{[a(b_2x)/b_2^2]
w(s,x)w_{xx}(s,x)\nonumber\\
&+&[\mu(b_2x)/b_2]w(s,x)w_x(s,x)\nonumber\\
 &+&[a(b_2x)/b_2^2-a(b_1x)/b_1^2]w(s,x)\theta_{xx}^{b_1}(t,x)\nonumber\\
  &+&w(s,x)[\mu(b_2x)/b_2-\mu(b_1x)/b_1]w(s,x)\theta_x^{b_1}(t,x)\big \}dxds\nonumber\\
&\equiv & E_1 +E_2+E_3+E_4.
\end{eqnarray}
We now estimate terms $E_i $, $i=1,\cdots, 4$, at both sides of
(\ref{425}) as follows. \vskip 5pt \noindent Firstly,
\begin{eqnarray}\label{426}
\int_0^{t}\int_0^1 w(s,x)w_t(s,x)dxds&=&\int_0^1
\frac{1}{2}w^2(t,x)dx.
\end{eqnarray}
Secondly, by Corollary A.1 in appendix and  definitions of $a(x)$
and $\mu (x) $, there exit positive constants $D_1 $, $D_2  $ and $
D_3 $ such that $[\mu(b_2z)/b_2]^2\leq D_1$, $[a(bx)/b^2]'\geq 0$,
$[a(b_2x)/b_2^2]\geq D_2$ and  $[a(b_2x)/b_2^2]'\leq D_3$, so by
Young's inequality, we have for any $\lambda_1>0$ and $\lambda_2>0$
\begin{eqnarray}\label{427}
E_1&=&\int_0^{t}\int_0^1[a(b_2x)/b_2^2]w(s,x)w_{xx}(s,x)dxds\nonumber\\
&=&-\int_0^{t}\int_0^1 [a(b_2x)/b_2^2]w_x^2(s,x)dxds\nonumber\\
&&-
\int_0^{t}\int_{0}^{m/b_2}[a(b_2x)/b_2^2]^{'}w_x(s,x)w(s,x)dxds\nonumber\\
&\leq & -D_2\int_0^{t}\int_0^1
w_x^2(s,x)dxds\nonumber\\
&& +D_3\int_0^{t}\int_0^1 [\lambda_1 w_x^2(s,x)+\frac{1}{4\lambda_1}
w^2(s,x)]dxds
\end{eqnarray}
and
\begin{eqnarray}\label{428}
E_2&=&\int_0^{t}\int_0^1[\mu(b_2x)/b_2]w(s,x)w_x(s,x)dxds\nonumber\\
&\leq
&\lambda_2\int_0^{t}\int_0^1 w_x^2(s,x)dxds\nonumber\\
&&+\frac{D_1}{4\lambda_2}\int_0^{t}\int_0^1 w^2(s,x)dxds.
\end{eqnarray}
Thirdly, it is easy to see from  Corollary A.1 in appendix that
$[a(bx)/b^2]$, $[a(bx)/b^2]'$and $[\mu(bx)/b]$ are Lipschitz
continuous for all $x\in(x_1/b_2,x_2/b_1)$, that is, there exists an
$L>0$ such that
\begin{eqnarray}\label{518}
\left\{
\begin{array}{l l l}
|[a(b_2x)/b_2^2]-[a(b_1x)/b_1^2]|\leq L|b_2-b_1|,\\
|[a(b_2x)/b_2^2]'-[a(b_1x)/b_1^2]'|\leq L|b_2-b_1|,\\
|[\mu(b_2x)/b_2]-[\mu(b_1x)/b_1]|\leq L|b_2-b_1|.
\end{array}\right.
\end{eqnarray}
 Noting that $E_3$ has  the following expressions:
\begin{eqnarray}\label{429}
E_3&=&\int_0^t\int_0^1 \{a(b_2x)/b_2^2-a(b_1x)/b_1^2\}w(s,x)\theta_{xx}^{b_1}(s,x)dxds\nonumber\\
&=&-\int_0^t\int_0^1 \{a(b_2x)/b_2^2-a(b_1x)/b_1^2\}w_x(s,x)\theta_{x}^{b_1}(s,x)dxds\nonumber\\
&&-\int_0^t\{\int_0^{m/b_2}\{a(b_2x)/b_2^2-a(b_1x)/b_1^2\}'w(s,x)\theta_{x}^{b_1}(s,x)dx\nonumber\\
&&-\int_{m/b_2}^{m/b_1}\{a(b_2x)/b_2^2-a(b_1x)/b_1^2\}'w(s,x)\theta_{x}^{b_1}(s,x)dx\nonumber\\
&&-\int_{m/b_1}^{1}\{a(b_2x)/b_2^2-a(b_1x)/b_1^2\}w(s,x)'\theta_{x}^{b_1}(s,x)dx\}ds\nonumber\\
&=&-\int_0^t\int_0^1 \{a(b_2x)/b_2^2-a(b_1x)/b_1^2\}w_x(s,x)\theta_{x}^{b_1}(s,x)dxds\nonumber\\
&&-\int_0^t\{\int_0^{m/b_2}\{a(b_2x)/b_2^2-a(b_1x)/b_1^2\}'w(s,x)\theta_{x}^{b_1}(s,x)dx\nonumber\\
&&-\int_{m/b_2}^{m/b_1}\{a(b_2x)/b_2^2-
a(b_1x)/b_1^2\}'w(s,x)\theta_{x}^{b_1}(s,x)\}dsdx\nonumber\\
&=& E_{31}+E_{32}+E_{33},
\end{eqnarray}
\begin{eqnarray*} \lim_{b_2\rightarrow b_1}\{|E_{33}|\}=0,
\end{eqnarray*}
and using (\ref{518}),
  by
the same way as in (\ref{427}) and (\ref{428}), we have for any
$\lambda_3>0$ and $\lambda_4>0$
\begin{eqnarray*}
E_{31}&=& -\int_0^t\int_{0}^{1}
\{a(b_2x)/b_2^2-a(b_1x)/b_1^2\}w(s,x)\theta_{x}^{b_1}(s,x)dxds\\
&\leq &\frac{L^2(b_2-b_1)^2}{4\lambda_3}\int_0^t\int_0^1
[\theta_{x}^{b_1}(s,x)]^2dxds\\
&& +\lambda_3\int_0^t\int_0^1 [w_x^2(s,x)+w^2(s,x)]dxds
\end{eqnarray*}
and
\begin{eqnarray*} E_{32}&=& -\int_0^t\int_{0}^{m/b_2}
\{a(b_2x)/b_2^2-a(b_1x)/b_1^2\}'w(s,x)\theta_{x}^{b_1}(s,x)dxds\\
&\leq &\frac{L^2(b_2-b_1)^2}{4\lambda_4}\int_0^t\int_0^1
[\theta_{x}^{b_1}(s,x)]^2dxds\\
&& +\lambda_4\int_0^t\int_0^1 [w_x^2(s,x)+w^2(s,x)]dxds.
\end{eqnarray*}
By Corollary A.1 in appendix, there exists a constant $D_4>0$ such
that $|[a(bz)/b^2]'-[\mu(bx)/b]|\leq D_4$ and
$\lambda_5=\inf\limits_{b_1\leq b\leq b_2}\{a(bx)/b^2\}>0$. Then, by
the boundary conditions, we estimate $\int_0^t\int_0^1
[\theta_{x}^{b}(s,x)]^2dxds$ for $b\in [b_1,b_2]$ as follows:
\begin{eqnarray*}
0&=&\int_0^t\int_0^1 \theta_{t}^{b}(s,x)\theta^{b}(s,x)\\
&& -[a(bx)/b^2]\theta_{xx}^{b}(s,x)\theta^{b}(s,x)
-[\mu(bx)/b]\theta_{x}^{b}(s,x)\theta^{b}(s,x)dxds\\
&=&\frac{1}{2}\int_0^1 [\theta^{b}(s,x)]^2dx
+\int_0^t\int_0^1 [a(bx)/b^2] [\theta_{x}^{b}(s,x)]^2dxds\\
&&+\int_0^t\int_0^1
[a(bx)/b^2]'[\theta_{x}^{b}(s,x)][\theta^{b}(s,x)]dxds\\
&&-\int_0^t\int_0^1 [\mu(bx)/b][\theta_{x}^{b}(s,x)][\theta^{b}(s,x)]dxds\\
&\geq & \lambda_5 \int_0^t\int_0^1 [\theta_{x}^{b}(s,x)]^2dxds
-\frac{\lambda_5}{2}\int_0^t\int_0^1 [\theta_{x}^{b}(s,x)]^2dxds\\
&&-\frac{1}{2\lambda_5}\int_0^t\int_0^1 [\theta^{b}(s,x)]^2dxds\\
&\geq &\frac{\lambda_5}{2}\int_0^t\int_0^1
[\theta_{x}^{b}(s,x)]^2dxds -\frac{D_4}{2\lambda_5}
\end{eqnarray*}
from which we see that
\begin{eqnarray*}
\int_0^t\int_0^1 [\theta_{x}^{b}(s,x)]^2dxds\leq
\frac{\lambda_5^2}{D_4}.
\end{eqnarray*}
Therefore we conclude that there exists a positive function
$B^{b_1}(b_2)$ such that
\begin{eqnarray*}
\lim\limits_{b_2\rightarrow b_1}B^{b_1}(b_2)=0
\end{eqnarray*}
and for $ 0\leq t\leq T$
\begin{eqnarray}\label{430}
E_3&=&\int_0^t\int_{0}^{1}\{a(b_2x)/b_2^2-a(b_1x)/b_1^2\}w(s,x)\theta_{xx}^{b_1}(t,x)dxds\nonumber\\
&\leq & B^{b_1}(b_2)+(\lambda_3 + \lambda_4 )\int_0^t\int_0^1
[w_x^2(s,x)+w^2(s,x)]dxds.
\end{eqnarray}
Finally, by using the same way as in estimating $E_3$, we can find a
positive function $B_1^{b_1}(b_2)$ such that
\begin{eqnarray*}
\lim\limits_{b_2\rightarrow b_1}B_1^{b_1}(b_2)=0
\end{eqnarray*}
and for any $\lambda_6 >0$
\begin{eqnarray}\label{431}
E_4&=&\int_0^t\int_0^1 \{\mu(b_2x)/b_2-\mu(b_1x)/b_1\}w(s,x)
\theta_x^{b_1}(t,x)dxds\nonumber\\
&\leq &  B_1^{b_1}(b_2)+\lambda_6\int_0^t\int_0^1 w^2(s,x)dxds.
\end{eqnarray}
By choosing  $\lambda_1$, $\lambda_2$, $\lambda_3$ and $\lambda_4$
such that $\lambda_2+ \lambda_1 D_3+ \lambda_3+ \lambda_4\leq D_2$ ,
it see from  (\ref{425}), (\ref{427})-(\ref{431}) that there exist
positive constants $C_1$ and $C_2$ such that
\begin{eqnarray*}
\int_0^1w^2(t,x)dx\leq
C_1\int_0^t\int_0^1w^2(s,x)dxds+C_2[B_1^{b_1}(b_2)+B^{b_1}(b_2)].
\end{eqnarray*}
By setting $F(t)=\int_0^t\int_0^1w^2(s,x)dxds$ and using the
Gronwall inequality, we get
\begin{eqnarray*}
F(t)\leq C_2[B_1^{b_1}(b_2)+B^{b_1}(b_2)]\exp\{C_1t\}.
\end{eqnarray*}
So
\begin{eqnarray*}
\lim\limits_{b_2\rightarrow
b_1}\int_0^t\int_0^1[\theta^{b_2}(s,x)-\theta^{b_1}(s,x)]^2dxds=0.
\end{eqnarray*}
Thus the proof has been done. \qquad $\Box$ \vskip 10pt\noindent
 \setcounter{equation}{0}
\section{ {\bf Appendix}}
 \vskip 10pt\noindent
The appendix lists the solutions of the two HJB equations and
properties of them. Since the procedure of solving the two equations
 is completely similar to that of Taksar  and Zhou\cite{ime04}(1998),
we omit it.
\begin{LA}\label{theorem31} Assume that
 $f(x)\in C^{2}$ satisfies the following HJB equation and boundary
conditions:
\begin{eqnarray}\label{31}
\left\{
\begin{array}{l l l}
\max\limits_{a\in[0,1]}[\frac{1}{2}\sigma^{2}a^{2}f^{''}(x)+(\mu-(1-
a)\lambda)f^{'}(x)-cf(x)]=0,\ 0\leq x\leq b_{0},\\
f^{'}(x)=1, \ \mbox{ for}\   x\geq b_{0},\\
f^{''}(x)=0, \ \mbox{ for} \  x\geq b_{0},\\
f(0)=0.
\end{array}\right.
\end{eqnarray}
(i) If $ \lambda \geq 2\mu $, then
\begin{eqnarray}\label{32} f(x)=\left\{
\begin{array}{l l l} f_1(x,b_0)=C_0(b_0)(e^{\zeta_{1}x}-e^{\zeta_{2}x}),&x\leq b_0,\\
f_2(x, b_0)=C_0(b_0)(e^{\zeta_{1}b_0}-e^{\zeta_{2}b_0})+x-b_0,
&x\geq b_0.
\end{array} \right.
\end{eqnarray}
If $ \mu <\lambda < 2\mu $, then
\begin{eqnarray}\label{33} f(x)=\left\{
\begin{array}{l l l}f_3(x,b_0)= \int_{0}^{x}X^{-1}(y)dy,\ x\leq m,\\
f_4(x, b_0)=\frac{C_1(b_0)}{\zeta_1}\exp{(\zeta_1(x-m))}
+\frac{C_2(b_0)}{\zeta_2}\exp{(\zeta_2(x-m))},\\ \ m<x<b_0, \\
f_5(x, b_0)=\frac{C_1(b_0)}{\zeta_1}\exp{(\zeta_1(b_0-m))}
+\frac{C_2(b_0)}{\zeta_2}\exp\{\zeta_2(b_0-m)\}\\
\qquad \qquad \quad + x-b_0, \ x\geq b_0.
\end{array} \right.
\end{eqnarray}
(ii)
\begin{eqnarray}\label{34}
\left\{
\begin{array}{l l l}
\max\mathcal {L}f(x)\leq 0\ \mbox{and }\  f^{'}(x)\geq 1 \
\mbox{for}\ x \geq
0,\\
f(0)=0,
\end{array}\right.
\end{eqnarray}
where   $\mathcal {L}= \frac{1}{2}\sigma^{2}a^{2}\frac{d^{2}}{d
x^{2}}+(\mu-(1- a)\lambda)\frac{d}{d x}-c$. \\
(iii) Let $A^*(x) $ is  the maximizer of the expression on the
left-hand side of (\ref{31}). \\
 If $ \lambda \geq 2\mu $, then  $A^*(x)=1 $
for $x\geq 0 $. If $ \mu <\lambda < 2\mu $, then
\begin{eqnarray}\label{35}
A^*(x)=A(x, b_0):=\left\{
\begin{array}{l l l}
-\frac{\lambda}{\sigma^2}(X^{-1}(x))X^{'}(X^{-1}(x)),&x\leq
m,\\1,&x>m,
\end{array} \right.
\end{eqnarray}
\vskip 10pt \noindent
 where $ X^{-1}$ denotes the inverse function
of $X(z)$.
\begin{eqnarray*}
&&\zeta_{1}=\frac{-\mu+\sqrt{\mu^2+2\sigma^2c}}{\sigma^2}, \quad
\zeta_{2}=\frac{-\mu-\sqrt{\mu^2+2\sigma^2c}}{\sigma^2},\\
&&b_0= 2\frac{\ln |\zeta_{2}/\zeta_{1}|}{\zeta_2-\zeta_1  },\quad
C_0(b_0)=\frac{1}{\zeta_1e^{\zeta_{1}b_0}-\zeta_2e^{\zeta_{2}b_0}},\Delta=b_0-m,\\
&&z_1=z_1(b_0)=\frac{\zeta_1-\zeta_2}
{(-\zeta_2-\lambda/\sigma^2)e^{\zeta_1\Delta}+(\zeta_1
+\lambda/\sigma^2)e^{\zeta_2\Delta}},\\
&& C_1(b_0)=z_1\frac{-\zeta_2-(\lambda/\sigma^2)}{\zeta_1-\zeta_2},
\quad
C_2(b_0)=z_1\frac{\zeta_1+(\lambda/\sigma^2)}{\zeta_1-\zeta_2},\\
&&C_3(b_0)=z_{1}^{1+c/\alpha}\frac{\lambda(c+\alpha(2\mu/\lambda-1))}
{2(\alpha+c)^2},\quad \alpha=\frac{\lambda^2}{2\sigma^2},\\
&&C_4(b_0) =-\frac{(\lambda-\mu)c}{(\alpha+c)^2}
+\frac{(\lambda-\mu)\alpha}{(\alpha+c)^2}\ln{C_3(b_0)}
+\frac{(\lambda-\mu)\alpha}{(\alpha+c)^2}\ln{\frac{(\alpha+c)^2}{(\lambda-\mu)c}},\\
&&X(z)=C_3(b_0)z^{-1-c/\alpha}+C_{4}(b_0)-\frac{\lambda-\mu}{\alpha+c}\ln{z},
\ \forall z>0,\quad m(b_0)= X(z_1).
\end{eqnarray*}
\end{LA}
\begin{LA}\label{theorem32}
Let $b>b_{0}$. Assume that $g\in C^{1}(R_+) \cap g\in
C^{2}(R_+\setminus \{b\})$ satisfies the following HJB equation and
boundary conditions:
\begin{eqnarray}\label{36}
\left\{
\begin{array}{l l l}
\max\limits_{a\in[0,1]}[\frac{1}{2}\sigma^{2}a^{2}g^{''}(x)+(\mu-(1-
a)\lambda)g^{'}(x)-cg(x)]=0,\ 0\leq x\leq b,\\
g^{'}(x)=1, \ \ \mbox{ for}\ \   x\geq b,\\
g^{''}(x)=0, \ \ \mbox{ for} \ \  x>b,\\
g(0)=0.
\end{array}\right.
\end{eqnarray}
(i) If $ \lambda \geq 2\mu $, then
\begin{eqnarray}\label{37} g(x)=\left\{
\begin{array}{l l l} f_1(x,b),&x\leq b,\\
f_2(x, b), &x\geq b.
\end{array} \right.
\end{eqnarray}
If $ \mu <\lambda < 2\mu $, then
\begin{eqnarray}\label{38} g(x)=\left\{
\begin{array}{l l l}f_3(x,b),\ x\leq m(b),\\
f_4(x, b), \ m(b)<x<b, \\
f_5(x, b), \ x\geq b.
\end{array} \right.
\end{eqnarray}
(ii)\begin{eqnarray}\label{39} \left\{
\begin{array}{l l l}
\max\mathcal {L}g(x)\leq 0,\  \mbox{for}\ \ x \geq
0,\\
g^{'}(x)\geq 1, \ \mbox{for}\ \  x \geq
b,\\
g(0)=0,
\end{array}\right.
\end{eqnarray}
where   $\mathcal {L}= \frac{1}{2}\sigma^{2}a^{2}\frac{d^{2}}{d
x^{2}}+(\mu-(1- a)\lambda)\frac{d}{d x}-c$. \\
(iii) Let $A^*(x) $ is  the maximizer of the expression on the
left-hand side of (\ref{36}).  \\  If $ \lambda \geq 2\mu $,then
$A^*(x)=1 $ for $x\geq 0 $. If $ \mu <\lambda < 2\mu $, then
\begin{eqnarray}\label{310}
A^*(x)=A(x, b).
\end{eqnarray}
{\sl Remark: Since $b\neq b_0 $, $g''(b)$ may not exist. We denote
by $g''(b)$ the  $\lim\limits_{x\uparrow b }g''(x) $ here.}
\end{LA}
\vskip 10pt \noindent As  direct consequences of Lemma A.1 and Lemma
A.2, we have the followings:
 \begin{CA}\label{corollary31}(i)
There exists a positive  constant $B=B(\lambda,\mu,c,\sigma )$,
which does not depend on $b $ and $x$, such that
\begin{eqnarray}\label{337}
\max\{|(A^*)'(x) |, |(A^*)''(x)    |\}\leq {B};
\end{eqnarray}
(ii) $1\geq A^*(x)\geq \min\{1,  \frac{2(\lambda-\mu )}{\lambda}
\}:=d>0$ for $x\geq 0$;\\ (iii) $A^*(x)$ is an increasing function
w.r.t. $x$.
\end{CA}
\begin{CA}\label{corollary32}
 For $b\geq b_0$, we have $\frac{\partial}{\partial b}g(b,x)\leq 0$
for $x\geq 0$.
\end{CA}
\vskip 10pt \noindent{\bf Acknowledgements.} This work is supported
by Projects 11071136 and 10771114 of NSFC, Project 20060003001 of
SRFDP, the SRF for ROCS, SEM  and the Korea Foundation for Advanced
Studies. We would like to thank the institutions for the generous
financial support. Special thanks also go to the participants of the
seminar stochastic analysis and finance at Tsinghua University for
their feedbacks and useful conversations. \noindent
 \setcounter{equation}{0}

\end{document}